\documentclass[conference,a4paper]{IEEEtran}
\usepackage[utf8]{inputenc}
\usepackage{verbatim}
\usepackage{amsmath, amsfonts}
\usepackage{amssymb}
\usepackage{mathrsfs}
\usepackage{graphicx}
\usepackage[numbers]{natbib}

\newtheorem{proposition}{Proposition}
\newtheorem{theorem}{Theorem}

\newtheorem{lemma}{Lemma}

\newcommand{\expectation}{\ensuremath{\mathbb{E}}}
\newcommand{\Expt}{\expectation}
\newcommand{\probability}{\ensuremath{\mathbb{P}}}
\newcommand{\Prob}{\probability}

\usepackage{url}

\begin{document}

\title{Conditions for Robustness of Polar Codes in the Presence of Channel Mismatch}
\author{\IEEEauthorblockN{Mine Alsan}\\
\small\IEEEauthorblockA{Information Theory Laboratory\\
Ecole Polytechnique F\' ed\' erale de Lausanne\\
CH-1015 Lausanne, Switzerland\\
Email: mine.alsan@epfl.ch}
\normalsize}
\maketitle
\pagestyle{empty}
\thispagestyle{empty}
\IEEEpeerreviewmaketitle

\maketitle

\begin{abstract}
A challenging problem related to the design of polar codes is ``robustness against channel parameter variations'' as stated in Ar{\i}kan's original work. 
In this paper, we describe how the problem of robust polar code design can be viewed as a mismatch decoding problem. We propose conditions 
which ensure a polar encoder/decoder designed for a mismatched B-DMC can be used to communicate reliably. 
In particular, the analysis shows that the original polar code construction method is robust over the class of binary symmetric channels.\\
\end{abstract}

\begin{IEEEkeywords}
Mismatched channels, channel polarization, polar codes.
\end{IEEEkeywords}

\section{Introduction}
In 2007, Ar{\i}kan~\cite{1669570} proposed polar codes as an appealing error correction method based on a phenomenon called channel polarization. 
This class of codes are proved to achieve the symmetric capacity of any binary discrete memoryless channel (B-DMC) 
using low complexity encoders and decoders, and their block error probability is shown to decrease exponentially in 
the square root of the blocklength~\cite{5205856}. 

Two basic channel transformations lie at the heart of channel polarization. Given a B-DMC $W: \mathcal{X}\rightarrow \mathcal{Y}$, the two successive channels characterized by these transformations $W^{-}: \mathcal{X} \rightarrow \mathcal{Y}^{2}$ and $W^{+}: \mathcal{X} \rightarrow \mathcal{Y}^{2}\times\mathcal{X}$ are defined by the following transition probabilities:
\small
\begin{align} 
  &W^{-}(y_1 y_2 \mid u_1) = \displaystyle\sum_{u_2\in\mathcal{X}} \frac{1}{2} W(y_1 \mid u_1 \oplus u_2) W(y_2 \mid u_2), \\
  &W^{+}(y_1 y_2 u_1\mid u_2) = \frac{1}{2} W(y_1 \mid u_1 \oplus u_2) W(y_2 \mid u_2). 
\end{align}
\normalsize
For a blocklength $N=2$, these channels would be indexed as $W_{2}^{(1)}$ and $W_{2}^{(2)}$. In general $N=2^n$, and the channels $W_{N}^{(i)}: \mathcal{X} \rightarrow \mathcal{Y}\times\mathcal{X}^{i-1}$, for $i=1, \ldots, N$, are synthesized by the recursive applications of these plus/minus transformations 
until sufficiently polarized, i.e. they are perfect or completely noisy channels.

The polarization idea is used to propose polar codes, and the recursive process leads to efficient encoding and decoding structures. 
On the encoder side, uncoded data bits are sent only through those perfect channels. For the rest, bits are fixed beforehand and revealed to the decoder as well. On the decoder side,
the synthesized channels lend themselves to a particular decoding procedure referred to as successive cancellation decoder (SCD).
At the i-th stage, on those good channels, the SCD estimates the channel input $u_{i}$ with law $W_{N}^{(i)}(y_{1}^{N} u_{1}^{i-1}\mid u_{i})$ according to maximum likelihood (ML) decision rule for the i-th channel using the previous estimates $\hat{u}_{1}^{i-1}$ 
and supplies the new estimate $\hat{u}_{i}$ to the next stages. The analysis carried in~\cite{1669570} shows that this SCD performs with vanishing error probability. 

A particular aspect of polar codes is that they are channel specific designs. The polarization process is adjusted to the particular channel at hand, whence the index set of the synthesized good channels.   
This set, referred as the information set $\mathcal{A}$, is required both by the encoder and decoder.  
The situation in which this knowledge is partially missing have been already addressed. Let $W$ and $V$ be two given B-DMCs. The following two cases are known to lead to an ordering $\mathcal{A_{V}} \subset \mathcal{A_{W}}$:
If $V$ is a binary erasure channel (BEC) with larger Bhattacharyya parameter than the channel $W$, or $V$ is a stochastically degraded version of $W$~\cite{1669570}. 
These results help the designer to use safely the information set designed for the channel $V$ for communication over $W$.

On the other hand, a critical point is the assumption of the availability of the channel knowledge at the decoder. Indeed, the described SCD not only requires the information set
but also the exact channel knowledge to function. Therefore, if the true channel is unknown, the code design, including the decoding rule, should be based on a mismatched channel \cite{394641}. 
In this work we assume the same SCD rule is kept, but instead of the true channel law a different one is employed in the decision procedure. 
We want to communicate reliably over the channel $W$ using the polar code designed for the mismatched channel $V$
(including the information set, encoder, and decoder), achieving rates up to the symmetric capacity of the mismatched channel $V$. 

The article follows with the preliminaries section, then we explore the results in the subsequent section, and the final section briefly discusses the results. 

\section{Preliminaries}
To assess the performance of mismatched polar codes,
we revisit expressions derived  in \cite{Mine:ISITA12} for the average probability of error under SCD with respect to a mismatched channel.
These derivations follow closely the matched counterparts in~\cite{1669570}. 

The SCD described in the introduction is closely tied to a channel splitting operation.  
After channel combining, the splitting synthesizes the 
channels whose transition probabilities are given by:
\begin{equation}
  W_{N}^{(i)}(y_{1}^{N}u_{1}^{i-1}|u_{i}) = \displaystyle\sum_{u_{i+1}^{N}} \displaystyle\frac{1}{2^{N-1}} W(y_{1}^{N}|u_{1}^{N}),
 \end{equation} 
 where $W(y_{1}^{N}|u_{1}^{N}) = \displaystyle\prod_{i=1}^{N} W(y_i|u_i)$.

We define the likelihood ratio (LR) of a given B-DMC $W$ as $L_{W}(y) = \displaystyle\frac{W(y|1)}{W(y|0)}$. 
Decision functions similar to ML decoding rule can then be defined as
\begin{equation}
d_{W}^{(i)}(y_{1}^{N}, \hat{u}_{1}^{i-1}) = \left\{\begin{array}{ll}
							0, & \hbox{if} \hspace{1mm} L_{W_{N}^{(i)}}\left(y_{1}^{N}, \hat{u}_{1}^{i-1}\right) < 1 \\
							1, & \hbox{if} \hspace{1mm} L_{W_{N}^{(i)}}\left(y_{1}^{N}, \hat{u}_{1}^{i-1}\right) > 1\\
							 * & \hbox{if} \hspace{1mm} L_{W_{N}^{(i)}}\left(y_{1}^{N}, \hat{u}_{1}^{i-1}\right) = 1
						       \end{array}  \right.,
\end{equation}
where $*$ is chosen from the set $\{0, 1\}$ by a fair coin flip.

The polar SCD will decode the received output in $N$ stages using a chain of estimators from $i = 1, \dots, N$ each depending on the previous ones. The estimators are defined as
\begin{equation}
\hat{u}_{i} =  \left\{\begin{array}{ll}
                                 u_{i}, & \hbox{if}  \hspace{3mm} i\in\mathcal{A}^{c} \\
				 d_{W}^{(i)}(y_{1}^{N}, \hat{u}_{1}^{i-1}), &\hbox{if}  \hspace{3mm} i\in\mathcal{A}
                                \end{array} \right..
\end{equation}

Let $Pe(W, V, \mathcal{A})$ denote the best achievable block error probability 
over the ensemble of all possible choices of the set $\mathcal{A}^{c}$ when $|\mathcal{A}| = \lfloor NR \rfloor$
under mismatched successive cancellation decoding with respect to the channel $V$ when the true channel is $W$. Then, one can show that
\begin{equation}\small
Pe(W, V, \mathcal{A}) 
\leq \displaystyle\sum_{i\in\mathcal{A}} Pe_{N}^{(i)}(W, V),
\end{equation}\normalsize
where $Pe_{N}^{(i)}(W, V)$ is defined as
\begin{multline}\small
\sum_{y_{1}^{N}, u_{1}^{N}} \displaystyle\frac{1}{2^N}W(y_{1}^{N}|u_{1}^{N}) \mathbf{1}\{\frac{V_{N}^{(i)}(y_{1}^{N}, u_{1}^{i-1}\mid u_{i} \oplus 1)}{V_{N}^{(i)}(y_{1}^{N}, u_{1}^{i-1}\mid u_{i})} > 1\} \\
+ \displaystyle\frac{1}{2} \sum_{y_{1}^{N}, u_{1}^{N}} \displaystyle\frac{1}{2^N}W(y_{1}^{N}|u_{1}^{N}) \mathbf{1}\{\frac{V_{N}^{(i)}(y_{1}^{N}, u_{1}^{i-1}\mid u_{i} \oplus 1)}{V_{N}^{(i)}(y_{1}^{N}, u_{1}^{i-1}\mid u_{i})} = 1\}  
\end{multline}\normalsize
with $\mathbf{1}\{.\}$ denoting the indicator function as usual.

For channels symmetrized under the same permutation, the next proposition can be proved using \cite[Corollary 1]{1669570}.
\begin{proposition}\label{prop::Pe_i_sym}
Let $W$ and $V$ be symmetric B-DMCs symmetrized under the same permutation. Then,
\begin{equation}\label{eq::Pe_i_sym}\small
Pe_{N}^{(i)}(W, V) = \displaystyle\sum_{y_{1}^{N}} W(y_{1}^{N}|0_{1}^{N}) \mathbf{H}\left(L_{V_{N}^{(i)}}\left(y_{1}^{N}, 0_{1}^{i-1}\right)\right), 
\end{equation}\normalsize
where $\mathbf{H}\left(L_{V_{N}^{(i)}}\left(y_{1}^{N}, 0_{1}^{i-1}\right)\right)$ is defined as the following sum
\begin{equation}
\mathbf{1}\{L_{V_{N}^{(i)}}\left(y_{1}^{N}, 0_{1}^{i-1}\right) > 1\} 
+ \displaystyle\frac{1}{2} \mathbf{1}\{L_{V_{N}^{(i)}}\left(y_{1}^{N}, 0_{1}^{i-1}\right) = 1\}.
\end{equation}
\end{proposition}

For shorthand notation we will use $L_{V_{N}^{(i)}}\left(y_{1}^{N}\right) \triangleq L_{V_{N}^{(i)}}\left(y_{1}^{N}, 0_{1}^{i-1}\right)$. 
The next proposition explores the recursive structure of the LR computations.

\begin{proposition}\cite{1669570}\label{prop::L_sym_func}
The LRs satisfy the recursion
\begin{align}
&L_{V_{2N}^{(2i-1)}}(y_{1}^{2N}) = \frac{L_{V_{N}^{(i)}}(y_{1}^{N}) + L_{V_{N}^{(i)}}(y_{N+1}^{2N})}{1 + L_{V_{N}^{(i)}}(y_{1}^{N}) L_{V_{N}^{(i)}}(y_{N+1}^{2N})}, \\
  &L_{V_{2N}^{(2i)}}(y_{1}^{2N}) = L_{V_{N}^{(i)}}(y_{1}^{N})L_{V_{N}^{(i)}}(y_{N+1}^{2N}). 
 \end{align}
Hence, the computed LRs can be seen as symmetric functions $f(L_{V_{N}^{(i)}}(y_{1}^{N}), L_{V_{N}^{(i)}}(y_{N+1}^{2N}))$ of the arguments.
 \end{proposition}

We will use the following notation
\begin{equation}
 \Prob_{W}\left[L_{V_{N}^{(i)}}\left(y_{1}^{N}\right) \geq 1 \right] = \displaystyle\sum_{y_1^N} W(y_{1}^{N}|0_{1}^{N}) \mathbf{1}\{L_{V_{N}^{(i)}}\left(y_{1}^{N}\right) \geq 1\}.
\end{equation}
Similar notation will hold for different sets considered within the indicator function. 
We will also use $\Expt_{W}\left[\mathbf{1}\{.\}\right] \triangleq \Prob_{W}\left[. \right]$ interchangeably.

Given two B-DMCs $W$ and $V$, we denote by $\mathbf{1}\{L_{V}(y) \geq 1\}^{W} \prec_{SD} \mathbf{1}\{L_{V}(y) \geq 1\}^{V}$ if the distribution of the random variable $\mathbf{1}\{L_{V}(y) \geq 1\}$
under the distribution $W(y|0)$ is stochastically dominated by the distribution under $V(y|0)$. 
For a definition of stochastic dominance, see for instance~\cite[Chapter 1.2, Theorem B]{SD:Def}.
By definition the condition implies 
\begin{equation}
 \Expt_{W}[F(\mathbf{1}\{L_{V}(y) \geq 1\})] \leq \Expt_{V}[F(\mathbf{1}\{L_{V}(y)\geq 1\})] 
\end{equation}
holds for any non-decreasing function $F(.)$. As an example, the cases where $W$ and $V$ are BSCs with crossover probabilities $\epsilon_W \leq \epsilon_V \leq 0.5$  
satisfy $\mathbf{1}\{L_{V}(y) \geq 1\}^{W} \prec_{SD} \mathbf{1}\{L_{V}(y) \geq 1\}^{V}$ order.
Similar notation will also be used for the $\mathbf{1}\{L_{V}(y) \leq 1\}$ random variable. 
\subsubsection*{Upper Bounds to $Pe_{N}^{(i)}(W, V)$}
We give two channel parameters which upper bound $Pe_{N}^{(i)}(W, V)$ for symmetric channels. The first one is simply $\Prob_{W}\left[L_{V_{N}^{(i)}}\left(y_{1}^{N}\right) \geq 1\right]$ when $W$ and $V$ are symmetrized under the same permutation.
The second parameter, analogous to the Bhattacharyya parameter defined for the matched scenario and referred to as the mismatched version of this quantity, is
$Z(W, V) = \displaystyle\sum_{y} W(y|0)\sqrt{L_{V}(y)}$. Extending the definition to the $i$-th synthesized channels, one can easily show that the bound
$Pe_{N}^{(i)}(W, V) \leq Z_{N}^{(i)}(W, V) \triangleq Z(W_{N}^{(i)}, V_{N}^{(i)})$ holds for symmetric channels. Naturally, $Pe(W, V)$ and $Z(W, V)$ will denote the parameters when $N=1$ and $i=1$. For the matched case, we will simply write $Pe_{N}^{(i)}(W)$ and $Z_{N}^{(i)}(W)$. 
 
\section{Results}\label{sec:results}
The next theorem states the main result of this paper. 

\begin{theorem}\label{thm::thm2}
Let $W$ and $V$ be two B-DMCs symmetrized under the same permutation which satisfy the following conditions:
\begin{enumerate}
 \item[(i)] $\Prob_{V}\left[L_{V}(y) \leq 1\right] \geq  \Prob_{V}\left[L_{V}(y) \geq 1\right]$,
 \item[(ii)] $\Prob_{W}\left[L_V(y) \geq 1\right] \leq \Prob_{V}\left[L_V(y) \geq 1\right]$,
 \item[(iii)] $\Prob_{W}\left[L_V(y) \leq 1\right] \geq \Prob_{V}\left[L_V(y) \leq 1\right]$.
\end{enumerate}
Then, for any given $N = 2^n$ with $n=1, 2,\ldots$ and any given $i=1, \ldots, N$, we have $Pe_{N}^{(i)}(W, V) \leq Z_{N}^{(i)}(V)$. 
Moreover, $Pe_{N}^{(i)}(W, V) \leq Pe_{N}^{(i)}(V)$ holds for $\forall i\in\mathcal{A}$.
\end{theorem}

Theorem \ref{thm::thm2} will be proved using the following lemma and the subsequent theorem. 

\begin{lemma}\label{lem::L_equals_one_supermartingale}
The process $\Prob_{W}\left[L_{V_N^{(i)}}(y_1^{2N}) = 1\right]$ is a bounded submartingale in $\left[0, 1\right]$ which converges almost surely to the values $\{0, 1\}$.
\end{lemma}
The proof of Lemma \ref{lem::L_equals_one_supermartingale} is given in the Appendix.

\begin{theorem}\label{thm::one_step_preservation_2}
Let $W$ and $V$ be B-DMCs such that for a given $N = 2^n$ with $n=0, 1, 2,\ldots$ and a given $i=1, \ldots, N$ the following conditions hold:
\begin{enumerate}
 \item[A)] $\Prob_{V}\left[L_{V_{N}^{(i)}}(y_{1}^{N}) \leq 1\right] \geq  \Prob_{V}\left[L_{V_{N}^{(i)}}(y_{1}^{N}) \geq 1\right]$,
 \item[B)] $\Prob_{W}\left[L_{V_{N}^{(i)}}(y_{1}^{N}) \geq 1\right] \leq  \Prob_{V}\left[L_{V_{N}^{(i)}}(y_{1}^{N}) \geq 1\right]$,
 \item[C)] $\Prob_{W}\left[L_{V_{N}^{(i)}}(y_{1}^{N}) \leq 1\right] \geq  \Prob_{V}\left[L_{V_{N}^{(i)}}(y_{1}^{N}) \leq 1\right]$.
\end{enumerate}
Then, the basic polarization transformations preserve the above three conditions in the sense that, at the next level, they hold for the $2i$-th and $2i-1$-th indices.  
\end{theorem} 

An entire section will be devoted to the proof of Theorem \ref{thm::one_step_preservation_2} 
after we prove Theorem \ref{thm::thm2}.
\begin{proof}[Proof of Theorem \ref{thm::thm2}]
Assume the conditions (i), (ii), and (iii) hold. Then by Theorem \ref{thm::one_step_preservation_2}, 
the conditions are preserved for the synthetic channels created by the polar transformations. 
Hence, for $\forall i= 1, \ldots, N$, we get
\begin{equation}
 \Prob_{W}\left[L_{V_{N}^{(i)}}(y_1^N) \geq 1\right] \leq \Prob_{V}\left[L_{V_{N}^{(i)}}(y_1^N) \geq 1\right].
\end{equation}
Knowing the bounds $Pe_{N}^{(i)}(W, V) \leq \Prob_{W}\left[L_{V_{N}^{(i)}}(y_1^N) \geq 1\right]$ 
(as we assumed $W$ and $V$ are symmetrized under the same permutation) and $\Prob_{V}\left[L_{V_{N}^{(i)}}(y_1^N) \geq 1\right] \leq Z(V_{N}^{(i)})$ apply,
the relation $Pe_{N}^{(i)}(W, V) \leq Z(V_{N}^{(i)})$ is proved.

On the other hand, Proposition \ref{lem::L_equals_one_supermartingale} shows that once channels are sufficiently polarized, either $\Prob_{W}\left[L_{V_N^{(i)}}(y_1^{2N}) = 1\right] \approx 1$ or 
$\Prob_{W}\left[L_{V_N^{(i)}}(y_1^{2N}) = 1\right] \approx 0$. Moreover, one can easily find that the first case lead to a completely noisy channel, 
and only the second case can lead to a perfect channel under a possibly mismatched decoding. As the inequalities
\begin{equation}
 Pe_{N}^{(i)}(W, V) \leq \Prob_{W}\left[L_{V_N^{(i)}}(y_1^{N}) \geq 1\right] \leq \Prob_{V}\left[L_{V_N^{(i)}}(y_1^{N}) \geq 1\right]
\end{equation}
hold, it turns out that, for those indices $i\in\mathcal{A}$ which correspond to the good channels' picked by the polar code designed for the channel $V$ so that $\Prob_{V}\left[L_{V_N^{(i)}}(y_1^{N}) = 1\right] \approx 0$, we have  
\begin{equation}
Pe_{N}^{(i)}(W, V) \leq Pe_{N}^{(i)}(V), \quad \forall i\in\mathcal{A}
\end{equation}
as claimed. This completes the proof of the theorem.
\end{proof}

\subsection{Proof of Theorem \ref{thm::one_step_preservation_2}}
We first introduce a set of propositions needed in the proof.   

\begin{proposition}\label{prop::LR_mass_preservation}
 For a symmetric B-DMC channel $V$ such that the condition
\begin{equation}
 \Prob_{V}\left[L_{V_{N}^{(i)}}\left(y_{1}^{N}\right) < 1\right] \geq \Prob_{V}\left[L_{V_{N}^{(i)}}\left(y_{1}^{N}\right) > 1\right]
\end{equation}
holds for a given $N=2^n$ with $n=0,1, \ldots$ and for a given $i=1, \ldots, N$, the basic polarization transformations preserve the inequality, i.e. for $j = 2i-1, 2i$, we have
\begin{equation}
 \Prob_{V}\left[L_{V_{2N}^{(j)}}\left(y_{1}^{2N}\right) < 1\right] \geq \Prob_{V}\left[L_{V_{2N}^{(j)}}\left(y_{1}^{2N}\right) > 1\right]. 
\end{equation} 
\end{proposition}
The proof of Proposition \ref{prop::LR_mass_preservation} is given in the Appendix.
\begin{proposition}\label{prop::P_diff_av}
 For B-DMCs $W$ and $V$, we have 
\begin{multline}\label{eq::P_diff_av}
  \Prob_{W}\left[L_{V_{2N}^{(i)}}(y_{1}^{2N}) \geq 1\right] - \Prob_{V}\left[L_{V_{2N}^{(i)}}(y_{1}^{2N}) \geq 1\right] \\
  = \displaystyle\sum_{y_{1}^{N}} \left[W(y_{1}^{N}|0_{1}^{N}) - V(y_{1}^{N}|0_{1}^{N}) \right] \times \\
\displaystyle\sum_{y_{N+1}^{2N}} \left[W(y_{N+1}^{2N}|0_{1}^{N}) + V(y_{N+1}^{2N}|0_{1}^{N}) \right] \mathbf{1}\{L_{V_{2N}^{(i)}}(y_{1}^{2N}) \geq 1\}.
\end{multline}
\end{proposition}
\begin{proof}[Proof of Proposition \ref{prop::P_diff_av}]
We develop the right hand side of Equation \eqref{eq::P_diff_av}
\begin{multline}
\displaystyle\sum_{y_{1}^{2N}} \left[W(y_{1}^{2N}|0_{1}^{2N}) - V(y_{1}^{2N}|0_{1}^{2N})\right] \mathbf{1}\{L_{V_{2N}^{(i)}}(y_{1}^{2N}) \geq 1\} \\
+\displaystyle\sum_{y_{1}^{2N}} W(y_{1}^{N}|0_{1}^{N})V(y_{N+1}^{2N}|0_{1}^{N}) \times \hspace{30mm}\\  \hspace{20mm}\mathbf{1}\{f(L_{V_{N}^{(i)}}(y_{1}^{N}), L_{V_{N}^{(i)}}(y_{N+1}^{2N}))  \geq 1\} \\
-\displaystyle\sum_{y_{1}^{2N}} W(y_{N+1}^{2N}|0_{1}^{N})V(y_{1}^{N}|0_{1}^{N}) \times \hspace{30mm}\\  \hspace{20mm}\mathbf{1}\{f(L_{V_{N}^{(i)}}(y_{N+1}^{2N}),L_{V_{N}^{(i)}}(y_{1}^{N}))  \geq 1\} \\
=\Prob_{W}\left[L_{V_{2N}^{(i)}}(y_{1}^{2N}) \geq 1\right] - \Prob_{V}\left[L_{V_{2N}^{(i)}}(y_{1}^{2N}) \geq 1\right],
\end{multline}
where we used the symmetry of the LR functions described in Proposition \ref{prop::L_sym_func}. 
\end{proof}

\begin{proposition}\label{prop::SD_preservation_binary}
For any B-DMCs $W$ and $V$, we have 
\begin{multline}
\mathbf{1}\{L_V(y) \geq 1\}^{W} \prec_{SD} \mathbf{1}\{L_V(y) \geq 1\}^{V} \quad \\
\hbox{iff} \quad \Prob_{W}\left[L_V(y) \geq 1\right] \leq \Prob_{V}\left[L_V(y) \geq 1\right],
\end{multline}
\begin{multline}
\mathbf{1}\{L_V(y) \leq 1\}^{W} \succ_{SD} \mathbf{1}\{L_V(y) \leq 1\}^{V} \quad \\
\hbox{iff} \quad \Prob_{W}\left[L_V(y) \leq 1\right] \geq \Prob_{V}\left[L_V(y) \leq 1\right].
\end{multline} 
\end{proposition}
\begin{proof}[Proof of Proposition \ref{prop::SD_preservation_binary}]
The proposition follows by noting the random variables with the indicator functions are binary valued, so for both cases the two conditions are equivalent. 
\end{proof}
\begin{proposition}\label{prop::1_increasing_average}
\small$\Expt_{W}\left[\mathbf{1}\{L_{V_{2N}^{(2i)}}(y_1^{2N}) \geq 1\} | \mathbf{1}\{L_{V_N^{(i)}}(y_1^{N}) \geq 1\}\right]$
$\left(\Expt_{W}\left[\mathbf{1}\{L_{V_{2N}^{(2i)}}(y_1^{2N}) \leq 1\} | \mathbf{1}\{L_{V_N^{(i)}}(y_1^{N}) \leq 1\}\right]\right)$\normalsize \hspace{2mm} function is non-decreasing in 
\small$\mathbf{1}\{L_{V_N^{(i)}}(y_1^{N}) \geq 1\}$ $\left(\mathbf{1}\{L_{V_N^{(i)}}(y_1^{N}) \leq 1\}\right)$\normalsize.
The function \small$\Expt_{W}\left[\mathbf{1}\{L_{V_{2N}^{(2i-1)}}(y_1^{2N}) \geq 1\} | \mathbf{1}\{L_{V_{N}^{(i)}}(y_1^{N}) \geq 1\}\right]$ 
$\left(\Expt_{W}\left[\mathbf{1}\{L_{V_{2N}^{(2i-1)}}(y_1^{2N}) \leq 1\} | \mathbf{1}\{L_{V_N^{(i)}}(y_1^{N}) \leq 1\}\right]\right)$\normalsize \hspace{2mm}however, is
non-decreasing in \small$\mathbf{1}\{L_{V_N^{(i)}}(y_1^{N}) \geq 1\}$ $\left(\mathbf{1}\{L_{V_N^{(i)}}(y_1^{N}) \leq 1\}\right)$\normalsize \hspace{1mm}if the following condition holds:
\begin{equation}\label{eq::W_order}
 \Prob_{W}\left[L_{V_N^{(i)}}(y_1^{N}) \leq 1 \right] \geq \Prob_{W}\left[L_{V_N^{(i)}}(y_1^{N}) \geq 1\right].
\end{equation}
\end{proposition}    
\begin{proof}[Proof of Proposition \ref{prop::1_increasing_average}]
The claims for the plus operations are trivial. For the minus operation, the claims follow by noting that
\begin{multline}
 \Expt\left[\mathbf{1}\{L_{V_{2N}^{(2i-1)}}(y_1^{2N}) \geq 1\} | \mathbf{1}\{L_{V_N^{(i)}}(y_1^{N}) \geq 1\} = 0\right] \\
=  \Expt\left[\mathbf{1}\{L_{V_{2N}^{(2i-1)}}(y_1^{2N}) \leq 1\} | \mathbf{1}\{L_{V_N^{(i)}}(y_1^{N}) \leq 1\} = 0\right] \\
= \Prob_{W}\left[L_{V_N^{(i)}}(y_{N+1}^{2N}) \geq 1\} \right], 
\end{multline} 
and both
\begin{multline}
 \Expt\left[\mathbf{1}\{L_{V_{2N}^{(2i-1)}}(y_1^{2N}) \geq 1\} | \mathbf{1}\{L_{V_N^{(i)}}(y_1^{N}) \geq 1\} = 1\right] \\
 \geq \Prob_{W}\left[L_{V_N^{(i)}}(y_{N+1}^{2N}) \leq 1\} \right], 
 \end{multline}
 \begin{multline}
 \Expt\left[\mathbf{1}\{L_{V_{2N}^{(2i-1)}}(y_1^{2N}) \leq 1\} | \mathbf{1}\{L_{V_N^{(i)}}(y_1^{N}) \leq 1\} = 1\right] \\
\geq \Prob_{W}\left[L_{V_N^{(i)}}(y_{N+1}^{2N}) \leq 1\} \right]. 
\end{multline}
So by symmetry of $y_{1}^N$ and $y_{N+1}^{2N}$ in the construction, the condition in \eqref{eq::W_order} is sufficient to prove the monotonicity claims.
\end{proof}
\begin{proof}[Proof of Theorem \ref{thm::one_step_preservation_2}]
\begin{enumerate}
 \item[A$^{\pm}$)] We know condition A is preserved by Proposition \ref{prop::LR_mass_preservation}. 	
 \item[B$^{\pm}$)] Using Proposition \ref{prop::P_diff_av} we get
\begin{multline}
\Prob_{W}\left[L_{V_{2N}^{(i)}}(y_{1}^{2N}) \geq 1\right] - \Prob_{V}\left[L_{V_{2N}^{(i)}}(y_{1}^{2N}) \geq 1\right] \\
= \displaystyle\sum_{y_{1}^{N}} \left[W(y_{1}^{N}|0_{1}^{N}) - V(y_{1}^{N}|0_{1}^{N}) \right] \times\\
\Expt_{W+V}\left[\mathbf{1}\{L_{V_{2N}^{(i)}}(y_1^{2N}) \geq 1\} | \mathbf{1}\{L_{V_{N}^{(i)}}(y_{1}^{N}) \geq 1\}\right],
\end{multline}
where we have defined  \small
\begin{multline}\label{eq::E_W+V}
 \Expt_{W+V}\left[\mathbf{1}\{L_{V_{2N}^{(i)}}(y_1^{2N}) \geq 1\} | \mathbf{1}\{L_{V_N^{(i)}}(y_1^{N}) \geq 1\}\right] = \\
\displaystyle\sum_{y_{N+1}^{2N}} \left[W(y_{N+1}^{2N}|0_{1}^{N}) + V(y_{N+1}^{2N}|0_{1}^{N}) \right] \mathbf{1}\{L_{V_{2N}^{(i)}}(y_{1}^{2N}) \geq 1\}.
\end{multline}\normalsize
Moreover by Proposition \ref{prop::SD_preservation_binary}, condition B implies that $\mathbf{1}\{L_{V_{N}^{(i)}}(y_{1}^{N}) \geq 1\}^{W} \prec_{SD} \mathbf{1}\{L_{V_{N}^{(i)}}(y_{1}^{N}) \geq 1\}^{V}$. 
So, we will be done if we show that the random variables defined in \eqref{eq::E_W+V} obtained after applying the polar 
transformations are both non-decreasing transformations in $\mathbf{1}\{L_{V_N^{(i)}}(y_1^{N}) \geq 1\}$. We consider the cases the expectations are taken under $W$ and $V$ separately.
For 
\begin{equation*}
\Expt_{V}\left[\mathbf{1}\{L_{V_{2N}^{(i)}}(y_1^{2N}) \geq 1\} | \mathbf{1}\{L_{V_N^{(i)}}(y_1^{N}) \geq 1\}\right],
\end{equation*}
we know by taking $W = V$ in Proposition \ref{prop::1_increasing_average} and by condition
A that this claim holds. For 
\begin{equation*}
\Expt_{W}\left[\mathbf{1}\{L_{V_{2N}^{(i)}}(y_1^{2N})\geq 1\} | \mathbf{1}\{L_{V_N^{(i)}}(y_1^{N})\geq 1\}\right], 
\end{equation*}
we know once again by Proposition \ref{prop::1_increasing_average} that this is always true for the plus transformation and is also true for the minus transformation if we have
\begin{equation}\label{eq::W_order_2}
 \Prob_{W}\left[L_{V_N^{(i)}}(y_{1}^N) \leq 1\right] \geq  \Prob_{W}\left[L_{V_N^{(i)}}(y_{1}^N) \geq 1\right].
\end{equation}
Now we show that \eqref{eq::W_order_2} holds. Taking the difference of the inequalities stated in conditions B and C, we get
\begin{multline}
\Prob_{W}\left[L_{V_N^{(i)}}(y_{1}^N) \leq 1\right] -  \Prob_{W}\left[L_{V_N^{(i)}}(y_{1}^N) \geq 1\right] \\
\geq \Prob_{V}\left[L_{V_N^{(i)}}(y_{1}^N) \leq 1\right] -  \Prob_{V}\left[L_{V_N^{(i)}}(y_{1}^N) \geq 1\right] \geq 0,
\end{multline}
where the non-negativity follows by condition A. 
\item[C$^{\pm}$)] The proof can be carried following similar steps as in part B$^{\pm}$ showing that the transformations defined by
$\Expt_{W+V}\left[\mathbf{1}\{L_{V_{2N}^{(i)}}(y_1^{2N}) \leq 1\} | \mathbf{1}\{L_{V_{N}^{(i)}}(y_{1}^{N}) \leq 1\}\right]$ are also non-decreasing in $\mathbf{1}\{L_{V_N^{(i)}}(y_1^{N}) \leq 1\}$
using Proposition \ref{prop::1_increasing_average}, condition A, and Equation \eqref{eq::W_order_2}.
\end{enumerate}
\end{proof}

It is useful to remark that for those B-DMCs $W$ and $V$ such that no output has a LR which equals to one, the assumptions (ii) and (iii) of Theorem \ref{thm::thm2} can be merged into a single initial condition as $Pe(W, V) \leq Pe(V)$.
Following this remark, we now study in Theorem \ref{thm::one_step_preservation}, the one step preservation properties related to the channel parameter $Pe_{N}^{(i)}$.

\begin{theorem}\label{thm::one_step_preservation}
Let $W$ and $V$ be B-DMCs symmetrized under the same permutation
such that for a given $N = 2^n$ with $n=0, 1, 2,\ldots$ and a given $i=1, \ldots, N$ the following conditions hold:
\begin{enumerate}
 \item[A)] $\Prob_{V}\left[L_{V_{N}^{(i)}}(y_{1}^{N}) < 1\right] \geq  \Prob_{V}\left[L_{V_{N}^{(i)}}(y_{1}^{N}) > 1\right]$,
 \item[B)] $Pe_{N}^{(i)}(W,V) - Pe_{N}^{(i)}(V) \leq 0$.
\end{enumerate}
Then, the minus polar transformation preserves these conditions.
On the other hand, while the plus transformation preserves condition $A$, condition $B$ may not be preserved in general. 
\end{theorem}

\subsection{Proof of Theorem \ref{thm::one_step_preservation}}
We first introduce two propositions needed in the proof. The proof of the propositions are  given in the Appendix.
\begin{proposition}\label{prop::Pe_diff_recursion}
The quantities $Pe_{2N}^{(i)}(W, V) - Pe_{2N}^{(i)}(V)$ can be recursively computed as
 \begin{multline}\label{eq::minus_Pe_diff_recursion}
  Pe_{2N}^{(2i-1)}(W, V) - Pe_{2N}^{(2i-1)}(V) \\
= \displaystyle\sum_{y_{1}^{N}} \left[W(y_{1}^{N}|0_{1}^{N}) - V(y_{1}^{N}|0_{1}^{N})\right]  \mathbf{H}\left(L_{V_{N}^{(i)}}(y_{1}^{N})\right) K_{N}, 
\end{multline}
where  
\begin{multline}\label{eq::KN_def}
K_{N} = \left(\displaystyle\sum_{\begin{subarray}{c}
                          y_{N+1}^{2N}: \\
			  L_{V_{N}^{(i)}}(y_{N+1}^{2N})< 1
                         \end{subarray}} \left[W(y_{N+1}^{2N}|0_{1}^{N}) + V(y_{N+1}^{2N}|0_{1}^{N})\right]\right. \\
-\left.\displaystyle\sum_{\begin{subarray}{c}
                          y_{N+1}^{2N}: \\
			  L_{V_{N}^{(i)}}(y_{N+1}^{2N}) > 1
                         \end{subarray}} \left[W(y_{N+1}^{2N}|0_{1}^{N}) + V(y_{N+1}^{2N}|0_{1}^{N})\right]\right),
\end{multline}  
and \small
\begin{multline}
Pe_{2N}^{(2i)}(W, V) - Pe_{2N}^{(2i)}(V) = \displaystyle\sum_{y_{1}^{2N}} \left[W(y_{1}^{N}|0_{1}^{N}) - V(y_{1}^{N}|0_{1}^{N}) \right] \times \\
\left[W(y_{N+1}^{2N}|0_{1}^{N}) + V(y_{N+1}^{2N}|0_{1}^{N}) \right] \mathbf{H}\left(L_{V_{N}^{(i)}}(y_{1}^{N})L_{V_{N}^{(i)}}(y_{N+1}^{2N})\right). 
\end{multline}\normalsize
\end{proposition}
\begin{proposition}\label{prop::Prob_W_ordering}
 Assume $W$ and $V$ are B-DMCs such that the conditions A and B of Theorem \ref{thm::one_step_preservation} hold for a given $N = 2^n$ with $n=0, 1, 2,\ldots$ and a given $i=1, \ldots, N$. Then,
\begin{equation}
 \Prob_{W}\left[L_{V_{N}^{(i)}}(y_{1}^{N}) < 1\right] \geq  \Prob_{W}\left[L_{V_{N}^{(i)}}(y_{1}^{N}) > 1\right]. 
\end{equation}
\end{proposition}

\begin{proof}[Proof of Theorem \ref{thm::one_step_preservation}]
\begin{enumerate}
 \item[A$^{\pm}$)]  We know condition A is preserved by Proposition \ref{prop::LR_mass_preservation}.  	 	
 \item[B$^{-}$)] For the minus transformation, we have by Proposition \ref{prop::Pe_diff_recursion}
	  \begin{multline}\label{eq::Pe_diff_minus}\small
	   Pe_{2N}^{(2i-1)}(W,V) - Pe_{2N}^{(2i-1)}(V) \\ 
= \left[Pe_{N}^{(i)}(W,V) - Pe_{N}^{(i)}(V)\right] K_{N}. 
	  \end{multline}\normalsize
Now, we claim that $K_N \geq 0$, from which the sign of $ Pe_{2N}^{(2i-1)}(W,V) - Pe_{2N}^{(2i-1)}(V) \leq 0$ follows. To prove the claim, note that 
by equation \eqref{eq::KN_def}, the constant $K_{N}$ equals to 
\begin{align}
&\Prob_{W}\left[L_{V_{N}^{(i)}}\left(y_{N+1}^{2N}\right) < 1\right] + \Prob_{V}\left[L_{V_{N}^{(i)}}\left(y_{N+1}^{2N}\right) < 1\right] \nonumber\\
- &\Prob_{W}\left[L_{V_{N}^{(i)}}\left(y_{N+1}^{2N}\right) > 1\right] - \Prob_{V}\left[L_{V_{N}^{(i)}}\left(y_{N+1}^{2N}\right) > 1\right].
\end{align}
Then, the non-negativity of $K_N$ follows by condition A and Proposition \ref{prop::Prob_W_ordering} which shows the conditions A and B imply
\begin{equation}
\Prob_{W}\left[L_{V_{N}^{(i)}}(y_{1}^{N}) < 1\right] \geq  \Prob_{W}\left[L_{V_{N}^{(i)}}(y_{1}^{N}) > 1\right].
\end{equation}
\item[B$^{+}$)] We give a counterexample: Let $W$ be a BSC of crossover probability $0.3$ and $V$ a symmetric B-DMC with $\mathcal{Y} = \{0, e, 1\}$ such that the LRs take the values $\{1/4, 1, 4 \}$ with probabilities $V(y|0) = \{0.4, 0.5, 0.1\}$, respectively. One can check that although conditions A and B are satisfied for $N=1$ and $i=1$, condition B fails to hold after the plus transformation for $N=2$ and $i=2$. 
\end{enumerate}
\end{proof} 

We saw in Theorem \ref{thm::one_step_preservation} that we need to impose some more constraints on the mismatch channel to be used if we want to ensure condition B is preserved under both transformations. 

Consider the mismatched Bhattacharyya parameter we defined as $Z(W, V) = \displaystyle\sum_{y} W(y|0)\sqrt{L_{V}(y)}$. After applying the plus polar transformation we get 
$Z_{2N}^{(2i)}(W, V) = Z_{N}^{(i)}(W, V)^2$ as in the matched case shown in \cite{1669570}. Therefore, we have
\begin{equation}
 Z_{N}^{(i)}(W, V) - Z_{N}^{(i)}(V) \leq 0 \Rightarrow Z_{2N}^{(2i)}(W, V) - Z_{2N}^{(2i)}(V) \leq 0.
\end{equation}

In the next theorem, we explore the possible connection of such a result with Theorem \ref{thm::one_step_preservation}.

\begin{theorem}\label{thm::order_Pe_Z}
Assume the channels $W$ and $V$ described in the hypothesis of Theorem \ref{thm::one_step_preservation} also satisfy the following conditions for any $N = 2^{n}$ with $n = 1, 2, \ldots$ and for any $i=1, \ldots, N$:
\begin{equation}\label{eq::Pe_Z_relation}
  Pe_{N}^{(i)}(W,V) - Pe_{N}^{(i)}(V) < 0 \quad \hbox{iff} \quad Z_{N}^{(i)}(W,V) - Z_{N}^{(i)}(V) < 0,
\end{equation}
Then, the condition $B$ of Theorem \ref{thm::one_step_preservation} is preserved under both polar transformations.
\end{theorem}

The theorem statement simply tells that if the Bhattacharyya upper bounds follow the same behavior as their $Pe_{N}^{(i)}$ counterparts;
which can occur if for instance they are sufficiently tight for both the matched and mismatched error probabilities at any level, 
then as long as we design the polar code for a mismatched channel $V$ such that $Pe(W,V) \leq Pe(V)$ is satisfied, we are safe to use the code over the channel $W$.  
Although Theorem \ref{thm::order_Pe_Z} provides a partial solution to the design problem, unfortunately it is non-constructive at this stage. 
We would need to study which channels could satisfy these type of constraints.

\section{Discussions}
We took a designer's perspective to analyze the performance of mismatched polar codes, and we identified in Theorem \ref{thm::thm2} conditions under which
the polar code designed using Ar{\i}kan's original construction method \cite{1669570} for a given B-DMC can be used reliably for a mismatched channel. 
Are these conditions $(i)$, $(ii)$, and $(iii)$ given in Theorem \ref{thm::thm2} terrestrial? 
We give a positive answer by showing the set of BSCs of crossover probabilities $\epsilon_W \leq \epsilon_V \leq 0.5$ satisfy the three conditions:
$(i)$ is equivalent to $1-\epsilon_V \geq \epsilon_V$, $(ii)$ is equivalent to $\epsilon_W \leq \epsilon_V$, and $(iii)$ to $1-\epsilon_W \geq 1 - \epsilon_V$.
As illustrated in this specific example the conditions are rather natural ones, and perhaps, they even hold for other specific class of channels.

The robustness of polar codes over BSCs have also been previously discussed in \cite{Mine:ISITA12}. 
Theorem 1 in \cite{Mine:ISITA12} shows that 
replacing the minus polar transformation with a specific approximation results in the LRs of the synthesized channels $W_{N}^{(i)}$ and $V_{N}^{(i)}$ to be ordered for each $i = 1, \ldots, N$ as
\begin{align}
&1 \leq \tilde{L}_{V_{N}^{(i)}}\left(y_{1}^{N}, u_{1}^{i-1} \right) \leq \tilde{L}_{W_{N}^{(i)}}\left(y_{1}^{N}, u_{1}^{i-1}\right), \\
\hbox{or} \quad &\tilde{L}_{W_{N}^{(i)}}\left(y_{1}^{N}, u_{1}^{i-1}\right) \leq \tilde{L}_{V_{N}^{(i)}}\left(y_{1}^{N}, u_{1}^{i-1}\right) \leq 1,
\end{align}
where the symbol $\tilde{\hspace{3mm}}$ indicates computations use the approximation. 
So, the decoder estimate for a given output realization 
will be identical whether the computations are performed with respect to the approximated LRs of the channel $W$ or the channel $V$. In this case, for any $i=1, \ldots, N$, $\tilde{Pe}_{N}^{(i)}(W, V) = \tilde{Pe}_{N}^{(i)}(W)$ holds as well. Although the decoder is completely robust, no theoretical analysis is provided to argue what rates can ultimately be achieved by a successive cancellation decoder using the approximate computations. On the other hand, here, a consequence of Theorem 1 is that the compound capacity \cite{article:compound} of the set of BSCs, i.e. the capacity of the worst BSC in the set, is achievable by the polar code designed for this worst channel.

\section{Acknowledgment}
The author would like to thank Emre Telatar for helpful discussions.
This work was supported by Swiss National Science Foundation under grant number 200021-125347/1. 

\section{Appendix}
In this Appendix we prove Lemma \ref{lem::L_equals_one_supermartingale} and Propositions \ref{prop::LR_mass_preservation}, \ref{prop::Pe_diff_recursion}, and \ref{prop::Prob_W_ordering}. \\

\begin{proof}[Proof of Lemma \ref{lem::L_equals_one_supermartingale}]
The boundedness claim is trivial. Let $L_1 = L_{V_N^{(i)}}(y_1^{N})$ and $L_2 = L_{V_N^{(i)}}(y_{N+1}^{2N})$ for simplicity. We first note that
\begin{equation}\label{eq::Prob_L_one_minus}
\Prob_{W}\left[\displaystyle\frac{L_1+L_2}{1+L_1L_2}= 1 \right] = 2\Prob_{W}\left[L = 1 \right] - \Prob_{W}\left[L = 1 \right]^2,
\end{equation}
\begin{equation}
 \Prob_{W}\left[L_1L_2 = 1 \right] \geq \Prob_{W}\left[L = 1 \right]^2
\end{equation}
where we used the fact that $\Prob_{W}\left[L = 1 \right] \triangleq \Prob_{W}\left[L_1 = 1 \right] = \Prob_{W}\left[L_2 = 1 \right]$. Therefore,
\begin{equation}
\Prob_{W}\left[L_1L_2 = 1 \right] + \Prob_{W}\left[\displaystyle\frac{L_1+L_2}{1+L_1L_2}= 1 \right] \geq 2\Prob_{W}\left[L_1 = 1 \right].
\end{equation}
This inequality proves the process is a submartingale. By general results on bounded martingales, we know the process converges almost surely~\cite{1669570}. One can complete the proof that the convergence is to the extremes in a similar fashion to the proof carried in~\cite[Proposition 9]{1669570} of the convergence to the extremes of the Bhattacharyya parameters' process associated to the polarization transformations. 
By using \eqref{eq::Prob_L_one_minus}, we have
\begin{multline}
\Expt_{\pm}\left[|\Prob_{W}\left[ L^{\pm} = 1 \right] - \Prob_{W}\left[ L = 1 \right] |\right] \\
\geq \frac{1}{2}\Prob_{W}\left[L = 1 \right] \left(1 - \Prob_{W}\left[L = 1 \right]\right),
\end{multline}
and when the left side of this inequality goes to zero, $\{0, 1\}$ are the only possible values $\Prob_{W}\left[L = 1 \right]$ can take. 
\end{proof}
\hspace{10mm}\\
\begin{proof}[Proof of Proposition \ref{prop::LR_mass_preservation}]
For simplicity we define $L_{V_{N}^{(i)}}\left(y_{1}^{N}\right) = L_1$, $L_{V_{N}^{(i)}}\left(y_{N+1}^{2N}\right) = L_2$, and omit the subscript in $\Prob_{V}$. 
Note that by symmetry in the construction of polar codes $\Prob\left[L_1 < 1\right] = \Prob\left[L_2 < 1\right]$.

For the plus transformation, we use a property following from the symmetry of the channels
\begin{equation}\label{eq::sym_cond}
W(y|0) = \displaystyle\frac{W(y|1)}{L(y)} \Rightarrow \Prob\left[L(y) = \ell\right] = \displaystyle\frac{1}{\ell} \Prob\left[L(y)= \displaystyle\frac{1}{\ell}\right].
\end{equation}
We define the following notations
\begin{align}
 \Prob\left[ L_1 \gneq 1 \right] &\triangleq \Prob\left[ L_1 > 1 \right] + \frac{1}{2} \Prob\left[ L_1 = 1 \right],\\
 \Prob\left[ L_1 \lneq 1 \right] &\triangleq \Prob\left[ L_1 < 1 \right] + \frac{1}{2} \Prob\left[ L_1= 1 \right].
\end{align}
Then, we have
\begin{align} 
&\Prob\left[L_1L_2 \lneq 1\right] \nonumber\\
= &\displaystyle\sum_{\ell_1 < 1}\sum_{\ell_2 < 1} \Prob\left[L_1 = \ell_1\right]\Prob\left[L_2 = \ell_2\right]   \nonumber\\
&\hspace{1mm}+ \displaystyle\sum_{\ell_1 < 1}\sum_{1 \leq \ell_2 < 1/\ell_1} \Prob\left[L_1 = \ell_1\right]\Prob\left[L_2 = \ell_2\right]  \nonumber\\
&\hspace{1mm}+ \displaystyle\sum_{\ell_1 \geq 1}\sum_{\ell_2 \leq 1/\ell_1} \Prob\left[L_1 = \ell_1\right]\Prob\left[L_2 = \ell_2\right]   \nonumber\\
&\hspace{1mm}- \frac{1}{2}\Prob\left[L_1  = 1\right]^2 \\
= &\Prob\left[L_1  < 1\right]^2 - \frac{1}{2}\Prob\left[L_1  = 1\right]^2  \nonumber\\
&\hspace{1mm}+ \displaystyle\sum_{\ell_1 > 1}\sum_{1 \leq \ell_2 < \ell_1} \ell_1\Prob\left[L_1= \ell_1\right]\Prob\left[L_2 = \ell_2\right]  \nonumber\\
&\hspace{1mm}+ \displaystyle\sum_{\ell_1 \geq 1}\sum_{\ell_2 \geq \ell_1} \ell_2\Prob\left[L_1 = \ell_1\right]\Prob\left[L_2 = \ell_2\right]  \\
=  &\Prob\left[L_1  < 1\right]^2 - \frac{1}{2}\Prob\left[L_1  = 1\right]^2  \nonumber\\ 
&\hspace{1mm}+\displaystyle\sum_{\ell_1 > 1}\sum_{1 < \ell_2 < \ell_1} \ell_1\Prob\left[L_1= \ell_1\right]\Prob\left[L_2 = \ell_2\right] \nonumber\\
&\hspace{1mm}+ \Prob\left[L_1  = 1\right]\displaystyle\sum_{\ell_1 > 1}\ell_1\Prob\left[L_1= \ell_1\right]  \nonumber\\
&\hspace{1mm}+ \displaystyle\sum_{\ell_1 > 1}\sum_{\ell_2 \geq \ell_1} \ell_2\Prob\left[L_1 = \ell_1\right]\Prob\left[L_2 = \ell_2\right] \nonumber\\
+ &\Prob\left[L_1  = 1\right]\displaystyle\sum_{\ell_2 > 1}\ell_2\Prob\left[L_1= \ell_2\right] + \Prob\left[L_1  = 1\right]^2  \\
= &\Prob\left[L_1  < 1\right]^2 + \frac{1}{4}\Prob\left[L_1  = 1\right]^2 \nonumber\\
&\hspace{1mm}+ \Prob\left[L_1  = 1\right]\Prob\left[L_1  < 1\right]  \nonumber\\
+ &\displaystyle\sum_{\ell_1 > 1}\sum_{\ell_2 > 1} \Prob\left[L_1 = \ell_1\right]\Prob\left[L_2 = \ell_2\right] \max\{\ell_1, \ell_2\} \\
=  &\Prob\left[L_1  \lneq 1\right]^2   \nonumber\\ 
+\label{eq::L_plus_less_halfequal_than_one}&\hspace{1mm}\displaystyle\sum_{\ell_1 \gneq 1}\sum_{\ell_2 \gneq 1} \Prob\left[L_1 = \ell_1\right]\Prob\left[L_2 = \ell_2\right] \max\{\ell_1, \ell_2\} 
\end{align}
where we abuse the notation to define (note the $\gneq$ sign in the summation index)
\begin{align}
&\displaystyle\sum_{\ell_1 \gneq 1}\sum_{\ell_2 \gneq 1} \Prob\left[L_1 = \ell_1\right]\Prob\left[L_2 = \ell_2\right] \max\{\ell_1, \ell_2\}  \nonumber\\
&= \displaystyle\sum_{\ell_1 > 1}\sum_{\ell_2 > 1} \Prob\left[L_1 = \ell_1\right]\Prob\left[L_2 = \ell_2\right] \max\{\ell_1, \ell_2\}  \nonumber\\
&+ \Prob\left[L_1  = 1\right]\displaystyle\sum_{\ell_2 > 1}\ell_2\Prob\left[L_1= \ell_2\right] + \frac{1}{4}\Prob\left[L_1  = 1\right]^2. 
\end{align}
In the same spirit, we define
\begin{align}
&\displaystyle\sum_{\ell_1 \gneq 1}\sum_{\ell_2 \gneq 1} \Prob\left[L_1 = \ell_1\right]\Prob\left[L_2 = \ell_2\right] \min\{\ell_1, \ell_2\}  \nonumber\\
= &\displaystyle\sum_{\ell_1 > 1}\sum_{\ell_2 > 1} \Prob\left[L_1 = \ell_1\right]\Prob\left[L_2 = \ell_2\right] \min\{\ell_1, \ell_2\}  \nonumber\\
&+ \Prob\left[L_1 = 1\right]\Prob\left[L_1 > 1\right] + \frac{1}{4}\Prob\left[L_1 = 1\right]^2,
\end{align}
and we note that
\begin{align}\label{eq::min_plus_max}
&\displaystyle\sum_{\ell_1 \gneq 1}\sum_{\ell_2 \gneq 1} \Prob\left[L_1 = \ell_1\right]\Prob\left[L_2 = \ell_2\right] \times \nonumber\\
&\hspace{30mm} \left(\max\{\ell_1, \ell_2\} + \min\{\ell_1, \ell_2\} \right) \nonumber\\
=&\displaystyle\sum_{\ell_1 \gneq 1}\sum_{\ell_2 \gneq 1} \Prob\left[L_1 = \ell_1\right]\Prob\left[L_2 = \ell_2\right] (\ell_1 + \ell_2) \\
=&2\Prob\left[L_1  \lneq 1\right]\Prob\left[L_1  \gneq 1\right].
\end{align}
As 
\begin{align}
1 &= \Prob\left[L_1L_2 \lneq 1\right] + \Prob\left[L_1L_2 \gneq 1\right] \\
&=  \left(\Prob\left[L_1\lneq 1\right] + \Prob\left[L_1 \gneq 1\right]\right)^2 \\
&= \Prob\left[L_1 \lneq 1\right]^2 + \Prob\left[L_1 \gneq 1\right]^2 \nonumber\\
&\hspace{30mm}+ 2 \Prob\left[L_1 \lneq 1\right]\Prob\left[L_1 \gneq 1\right]
\end{align}
must hold, we get
\begin{align}
&\Prob\left[L_1L_2 \gneq 1\right] = \Prob\left[L_1 \gneq 1\right]^2  \nonumber\\
\label{eq::L_plus_greater_halfequal_than_one} &+ \displaystyle\sum_{\ell_1 \gneq 1}\sum_{\ell_2 \gneq 1} \Prob\left[L_1 = \ell_1\right]\Prob\left[L_2 = \ell_2\right] \min\{\ell_1, \ell_2\}. 
\end{align}
Therefore, \eqref{eq::L_plus_less_halfequal_than_one} and \eqref{eq::L_plus_greater_halfequal_than_one} proves that
\begin{equation}
\Prob\left[L_1L_2 < 1\right] \geq \Prob\left[L_1L_2 > 1\right]
\end{equation}
holds as claimed. For the minus transformation, we have
\begin{equation}
\Prob\left[ \displaystyle\frac{L_1 + L_2}{1 + L_1 L_2} < 1\right] 
= \Prob\left[L_1 < 1 \right]^{2} + \Prob\left[L_1 > 1\right]^{2},
\end{equation}
\begin{equation}
\Prob\left[ \displaystyle\frac{L_1 + L_2}{1 + L_1 L_2} > 1\right] 
= 2 \Prob\left[L_1 < 1\right] \Prob\left[L_1 > 1\right]. 
\end{equation}
By noting that the difference of these equals
\begin{equation}
\left(\Prob\left[L_1 < 1\right] - \Prob\left[L_1 > 1\right]\right)^{2} \geq 0,
\end{equation}
the claim for the minus transformation is proved.
\end{proof}  
\begin{proof}[Proposition \ref{prop::Pe_diff_recursion}]
First note that for symmetric B-DMCs $W$ and $V$ symmetrized under the same permutation, we have 
\begin{multline}\label{eq::Pe_diff_av}
  Pe_{2N}^{(i)}(W, V) - Pe_{2N}^{(i)}(V) = \displaystyle\sum_{y_{1}^{N}} \left[W(y_{1}^{N}|0_{1}^{N}) - V(y_{1}^{N}|0_{1}^{N}) \right] \times \\
 \hspace{2mm}\left(\displaystyle\sum_{y_{N+1}^{2N}} \left[W(y_{N+1}^{2N}|0_{1}^{N}) + V(y_{N+1}^{2N}|0_{1}^{N}) \right] \mathbf{H}\left(L_{V_{2N}^{(i)}}(y_{1}^{2N})\right)\right),
\end{multline}
which can be proved similarly to Proposition \eqref{prop::P_diff_av}. For simplicity we define $L_{V_{N}^{(i)}}\left(y_{1}^{N}\right) = L_1$, $L_{V_{N}^{(i)}}\left(y_{N+1}^{2N}\right) = L_2$. First observe that 
\begin{equation}
 \mathbf{H}\left(\frac{L_1 + L_2}{1 + L_1 L_2} \right) \\
 = \left\{\begin{array}{ll}
      \displaystyle\frac{1}{2}, &\hbox{if}\quad L_1 = 1,\\
				&\hbox{or}\quad L_2 = 1 \\
			     1, &\hbox{if}\quad L_1 < 1\quad\hbox{and}\quad L_2 > 1,\\
				&\hbox{or}\quad L_1 > 1\quad\hbox{and}\quad L_2 < 1 \\
			     0, &\hbox{if}\quad L_1 < 1\quad\hbox{and}\quad L_2 < 1,\\ 
				&\hbox{or}\quad L_1 > 1\quad\hbox{and}\quad L_2 > 1
	  \end{array} \right. .
\end{equation}
Then, we have
\begin{align}
&Pe_{2N}^{(2i-1)}(W, V) - Pe_{2N}^{(2i-1)}(V) \nonumber\\
= &\displaystyle\sum_{ \begin{subarray}{c}
			  y_{1}^{N}: \\
			  L_1 = 1
                    \end{subarray} } \left[W(y_{1}^{N}|0_{1}^{N}) - V(y_{1}^{N}|0_{1}^{N})\right] \times 1 \nonumber\\
&+ \displaystyle\sum_{ \begin{subarray}{c}
			  y_{1}^{N}: \\
			  L_1 > 1
                    \end{subarray} } \left[W(y_{1}^{N}|0_{1}^{N}) - V(y_{1}^{N}|0_{1}^{N})\right] \times \nonumber\\
&\hspace{10mm}
 \displaystyle\sum_{\begin{subarray}{c}
                           y_{N+1}^{2N}: \\
 			  L_2 \leq 1
                          \end{subarray}} \left[W(y_{N+1}^{2N}|0_{1}^{N}) + V(y_{N+1}^{2N}|0_{1}^{N})\right] \nonumber\\
&+ \displaystyle\sum_{ \begin{subarray}{c}
			  y_{1}^{N}: \\
			  L_1 < 1
                    \end{subarray} } \left[W(y_{1}^{N}|0_{1}^{N}) - V(y_{1}^{N}|0_{1}^{N})\right] \times \nonumber\\
&\hspace{10mm}
 \displaystyle\sum_{\begin{subarray}{c}
                           y_{N+1}^{2N}: \\
 			  L_2 \geq 1
                          \end{subarray}} \left[W(y_{N+1}^{2N}|0_{1}^{N}) + V(y_{N+1}^{2N}|0_{1}^{N})\right] \\
= &\displaystyle\sum_{ \begin{subarray}{c}
			  y_{1}^{N}: \\
			  L_1 = 1
                    \end{subarray} } \left[W(y_{1}^{N}|0_{1}^{N}) - V(y_{1}^{N}|0_{1}^{N})\right] \times 1 \nonumber\\
&+ \displaystyle\sum_{ \begin{subarray}{c}
			  y_{1}^{N}: \\
			  L_1 > 1
                    \end{subarray} } \left[W(y_{1}^{N}|0_{1}^{N}) - V(y_{1}^{N}|0_{1}^{N})\right] \times \nonumber\\
&\hspace{10mm}
 \displaystyle\sum_{y_{N+1}^{2N}} \left[W(y_{N+1}^{2N}|0_{1}^{N}) + V(y_{N+1}^{2N}|0_{1}^{N})\right] \mathbf{\overline{H}}\left(L_2 \right) \nonumber\\
&+ \displaystyle\sum_{ \begin{subarray}{c}
			  y_{1}^{N}: \\
			  L_1 < 1
                    \end{subarray} } \left[W(y_{1}^{N}|0_{1}^{N}) - V(y_{1}^{N}|0_{1}^{N})\right] \times \nonumber\\
&\hspace{10mm}
 \displaystyle\sum_{y_{N+1}^{2N}} \left[W(y_{N+1}^{2N}|0_{1}^{N}) + V(y_{N+1}^{2N}|0_{1}^{N})\right] \mathbf{H}\left(L_2 \right).
\end{align}
where the ``complement'' function of $\mathbf{H}$ is defined as
\begin{equation}
\mathbf{\overline{H}}\left(L_1\right) \triangleq \mathbf{1}\{L_1 < 1\} 
+ \displaystyle\frac{1}{2} \mathbf{1}\{L_1 = 1\}.
\end{equation}
By substituting $\mathbf{\overline{H}}\left(L_1\right) = 1 - \mathbf{H}\left(L_1\right)$ and regrouping the terms, we obtain
\begin{multline}
Pe_{2N}^{(2i-1)}(W, V) - Pe_{2N}^{(2i-1)}(V) \\
= \displaystyle\sum_{y_{1}^{N}} \left[W(y_{1}^{N}|0_{1}^{N}) - V(y_{1}^{N}|0_{1}^{N})\right] 2\mathbf{H}\left(L_1\right) \\
+ \displaystyle\sum_{y_{1}^{N}} \left[W(y_{1}^{N}|0_{1}^{N}) - V(y_{1}^{N}|0_{1}^{N})\right] \left[1 - 2\mathbf{H}\left(L_1\right)\right] \times \\
\displaystyle\sum_{y_{N+1}^{2N}} \left[W(y_{N+1}^{2N}|0_{1}^{N}) + V(y_{N+1}^{2N}|0_{1}^{N})\right] \mathbf{H}\left(L_2\right),
\end{multline}
where we used the fact that $1 - 2\mathbf{H}\left(L_1\right) = \mathbf{1}\{L_1 < 1\} - \mathbf{1}\{L_1 > 1\}$.
Now, note that the term in the second summation with the $1$ sums to 0. Hence, we get
\begin{multline}
Pe_{2N}^{(2i-1)}(W, V) - Pe_{2N}^{(2i-1)}(V) \\
= \displaystyle\sum_{y_{1}^{N}} \left[W(y_{1}^{N}|0_{1}^{N}) - V(y_{1}^{N}|0_{1}^{N})\right] \mathbf{H}\left(L_1 \right) \times \\
\left[2 - \displaystyle\sum_{y_{N+1}^{2N}} \left[W(y_{N+1}^{2N}|0_{1}^{N}) + V(y_{N+1}^{2N}|0_{1}^{N})\right] 2\mathbf{H}\left(L_2 \right) \right] \\
= \displaystyle\sum_{y_{1}^{N}} \left[W(y_{1}^{N}|0_{1}^{N}) - V(y_{1}^{N}|0_{1}^{N})\right] \mathbf{H}\left(L_1 \right) \times \\
\displaystyle\sum_{y_{N+1}^{2N}} \left[W(y_{N+1}^{2N}|0_{1}^{N}) + V(y_{N+1}^{2N}|0_{1}^{N})\right] \left[1 - 2\mathbf{H}\left(L_2 \right) \right]. 
\end{multline}
We recover Equation \eqref{eq::minus_Pe_diff_recursion} upon noticing $K_N$ defined in \eqref{eq::KN_def} equals 
\begin{equation}
\displaystyle\sum_{y_{N+1}^{2N}} \left[W(y_{N+1}^{2N}|0_{1}^{N}) + V(y_{N+1}^{2N}|0_{1}^{N})\right] \left[1 - 2\mathbf{H}\left(L_2 \right)\right]
\end{equation}
as $1 - 2\mathbf{H}\left(L_2 \right) = \mathbf{1}\{L_2 < 1\} - \mathbf{1}\{L_2 > 1\}$. This proves the claim for the minus transformation. The claim for the plus transformation can be obtained directly by the expression given in \eqref{eq::Pe_diff_av}.
\end{proof}

\begin{proof}[Proof of Proposition \ref{prop::Prob_W_ordering}]
We have
\begin{align}
 &\Prob_{W}\left[L(y_{1}^{N}) > 1\right] + \frac{1}{2}\Prob_{W}\left[L(y_{1}^{N}) = 1\right] \nonumber\\  &\hspace{10mm}-\Prob_{V}\left[L(y_{1}^{N}) > 1\right] - \frac{1}{2}\Prob_{V}\left[L(y_{1}^{N}) = 1\right] \nonumber\\
= &\Prob_{V}\left[L(y_{1}^{N}) < 1\right] + \frac{1}{2}\Prob_{V}\left[L(y_{1}^{N}) = 1\right] \nonumber\\  &\hspace{10mm}-\Prob_{W}\left[L(y_{1}^{N}) < 1\right] - \frac{1}{2}\Prob_{W}\left[L(y_{1}^{N}) = 1\right] \leq 0,
\end{align}
where the negativity follows by condition B. Therefore, adding both sides gives
\begin{multline}
 \Prob_{W}\left[L(y_{1}^{N}) > 1\right] - \Prob_{V}\left[L(y_{1}^{N}) > 1\right] \\ + \Prob_{V}\left[L(y_{1}^{N}) < 1\right] - \Prob_{W}\left[L(y_{1}^{N}) < 1\right] \leq 0.
\end{multline}
Hence,
\begin{multline}
 \Prob_{W}\left[L(y_{1}^{N}) < 1\right] - \Prob_{W}\left[L(y_{1}^{N}) > 1\right] \\
 \geq \Prob_{V}\left[L(y_{1}^{N}) < 1\right] - \Prob_{V}\left[L(y_{1}^{N}) > 1\right] \geq 0,
\end{multline}
where the non-negativity follows by condition A.
\end{proof}


\footnotesize
\begin{thebibliography}{10}

\bibitem{1669570}
E.~Ar{\i}kan, ``Channel polarization: a method for constructing capacity-achieving codes for symmetric binary-input memoryless channels," \emph{IEEE Trans. Inf. Theor.}, vol. 55, no. 7, pp. 3051-3073, 2009. 

\bibitem{5205856}
E.~Ar{\i}kan, and E.~Telatar, ``On the Rate of Channel Polarization,"  \emph{IEEE International Symposium on Information Theory (ISIT)}, pp.1493-1495, 2009.

\bibitem{394641}
I.~Csisz{\'a}r, and P.~Narayan, ``Channel capacity for a given decoding metric,'' \emph{IEEE Trans. Inf. Theor.}, vol. 41, no. 1, pp. 35 -43, 1995. 

\bibitem{Mine:ISITA12}
M.~Alsan, ``Performance of mismatched polar codes over BSCs'', \emph{International Symposium on Information Theory and its Applications (ISITA2012)}, 2012.


\bibitem{SD:Def}
R.~Szekli,``Stochastic Ordering and Dependence in Applied Probability'',  \emph{Lecture Notes in Statistics}, Springer-Verlag, 1995.

\bibitem{article:compound}
D.~Blackwell, and L.~Breiman, and A.J.~Thomasian, ``The capacity of a class of channels," \emph{The Annals of Mathematical Statistics}, vol. 3, no. 4, pp. 1229-1241, 1959. 

\end{thebibliography}
\normalsize
\end{document}